\documentclass[journal, doublecolumn]{IEEEtran}
\usepackage{amssymb}
\usepackage{amsthm}
\usepackage{graphicx}
\usepackage{algorithm}
\usepackage{algorithmic}
\usepackage{cite}
\usepackage{bm}
\usepackage{float}
\usepackage{multirow}
\usepackage{color}
\usepackage{subfigure}
\usepackage{caption}
\usepackage{subcaption}
\usepackage{epstopdf}
\usepackage{hyperref}
\usepackage{fixltx2e}
\usepackage{longtable}
\usepackage{diagbox}
\usepackage{changepage}
\usepackage{amsmath, amsfonts}
\usepackage{dsfont}
\usepackage{makecell}
\usepackage{optidef}

\theoremstyle{definition}

\theoremstyle{remark}

\theoremstyle{plain}
\newtheorem{theorem}{Theorem}
\newtheorem{lemma}{Lemma}
\newtheorem{proposition}{Proposition}
\newtheorem{corollary}{Corollary}

\begin{document}

\title{Physical Layer Security Performance of Pinching-Antenna Systems With In-Waveguide Attenuation}

\author{Xiaochen~Zhang,
        Haitao~Du,
        Yanyu~Cheng,
        Yushen~Lin,
        and~Kah~Chan~Teh

\thanks{Xiaochen Zhang is with School of Continuing and Lifelong Education, National University of Singapore, Singapore 119077 (e-mail: e1553773@u.nus.edu).

Haitao Du is with the School of Cyberspace, Hangzhou Dianzi University, Hangzhou 310018, China (e-mail: 242270031@hdu.edu.cn).

Yanyu Cheng, Yushen Lin, and Kah Chan Teh are with the School of Electrical and Electronic Engineering, Nanyang Technological University, Singapore 639798 (e-mail: ycheng022@e.ntu.edu.sg; yushen.lin@ntu.edu.sg; ekcteh@ntu.edu.sg).
}}

\maketitle

\begin{abstract}
Pinching antenna (PA) systems have recently gained significant attention.
While their physical-layer security (PLS) is being explored, 
most studies rely on idealized lossless models, ignoring practical waveguide attenuation.
In this paper, we investigate the PLS performance of PA systems under a more realistic attenuation-incorporated waveguide model. 
Specifically, we investigate a PA system-based secure communication scenario consisting of a base station (BS), a legitimate user, and a passive eavesdropper.
We derive expressions for closed-form upper and lower bounds on both the secrecy outage probability (SOP) and ergodic secrecy capacity (ESC). 
The results indicate that the PA system outperforms conventional fixed-antenna systems.
\end{abstract}

\begin{IEEEkeywords}
In-waveguide attenuation, physical-layer security, pinching antenna systems.
\end{IEEEkeywords}

\IEEEpeerreviewmaketitle

\section{Introduction}

In recent years, flexible antenna systems have attracted significant attention due to their potential for wireless channel reconfiguration. 
Representative architectures include reconfigurable intelligent surfaces (RIS), movable antenna systems, and fluid antenna systems \cite{sun_2025_physical}. 
By exploiting reconfigurable degrees of freedom, these architectures can effectively shape the wireless channel. Specifically, RIS consists of low-cost reflecting elements that dynamically reshape the propagation channel and establish virtual line-of-sight (LoS) links \cite{ding_2024_flexibleantenna}. 
However, RIS-assisted communication often suffers from severe cascaded path loss \cite{ding_2024_flexibleantenna}. 
In contrast, movable and fluid antenna systems enhance signal quality by adjusting antenna positions to exploit favorable channel conditions. 
However, to avoid prohibitive training overhead, channel estimation is typically restricted to small-scale position adjustments\cite{zhu_2023_movable, wong_2021_fluid}. 
In this context, PA systems built on dielectric waveguides have emerged as a promising architecture, offering large spatial flexibility and waveguide-enabled reconfiguration for channel customization.

The PA system was first prototyped by NTT DOCOMO in 2021 \cite{_2021_pinching}, employing a dielectric waveguide to guide electromagnetic signals. 
By pinching dielectric particles onto the waveguide, radio signals can be radiated from arbitrary positions \cite{_2021_pinching}. 
Unlike conventional flexible architectures, the waveguide in PA systems can extend over long distances, allowing the radiating point to move much closer to the user. 
Consequently, PA  can establish stable LoS links and mitigate large-scale path loss \cite{liu_2025_pinchingantenna}.
Existing studies on PA systems have primarily focused on legitimate-link performance. 
For instance, \cite{dimitriostyrovolas_2025_performance} established a performance analysis framework considering free space path loss and in-waveguide attenuation, deriving closed-form expressions for the SOP and average rate.
Furthermore, recent work has conducted a comprehensive evaluation of PA systems under both orthogonal multiple access (OMA) and non-orthogonal multiple access (NOMA) schemes, demonstrating that PA systems significantly reduce the SOP for LoS users compared to conventional antenna systems \cite{cheng_2025_on}. 
Other works have explored enhancement directions, such as array gain characterization
and beamforming optimization \cite{ouyang_2025_array, 
wang_2025_modeling} to improve reliability and resource efficiency.

However, wireless propagation is inherently open, and improving legitimate link performance alone is insufficient to guarantee security, as potential eavesdroppers may also exploit the PA systems' design.
Therefore, analyzing PA systems from a PLS perspective is essential. 
Some works have investigated secrecy enhancement through cooperation among multiple PAs to maximize the secrecy rate \cite{wang_2026_pinchingantenna}. 
Meanwhile, \cite{badarnehosamahs_2025_physicallayer} evaluates secrecy performance in a wiretap channel scenario and derives analytical expressions for secrecy capacity and SOP.

\subsection{Motivation and Contributions}
Existing studies often assume idealized lossless waveguide models, 
ignoring the practical in-waveguide attenuation that can significantly affect the secrecy performance and optimal antenna placement \cite{xu_2025_pinchingantenna}. 
This gap motivates the present work to evaluate PA systems under realistic, attenuation-aware conditions.
The contributions can be summarized as follows.
\begin{itemize}
\item By exploiting statistical geometry to model the link distances using probability density functions (PDFs) and cumulative distribution functions (CDFs), we derive closed-form expressions for the legitimate and eavesdropping channels.
\item We derive the closed-form expressions for upper and lower bounds for the SOP and ESC, which facilitate further studies such as the optimization of PA deployment strategies.
\item To obtain the diversity order and high signal-to-noise ratio (SNR) slope, we further derive asymptotic expressions of the SOP and ESC in the high-SNR regime.
\item Numerical results demonstrate that the proposed PA systems achieve superior secrecy performance compared to conventional FA systems, despite accounting for in-waveguide attenuation.
\end{itemize}

\section{System Model}
We consider a PA communication system with a BS, a legitimate user $\mathrm{U}_b$, and an eavesdropper $\mathrm{U}_w$ in a room with side length $D$, as shown in Fig.~\ref{fig:system_model}.
To conveniently describe the system model, we establish a three-dimensional Cartesian coordinate system.
$\mathrm{U}_b$ and $\mathrm{U}_w$ are located at $\psi_b = [x_1, y_1, 0]$ and $\psi_w = [x_2, y_2, 0]$, respectively. 
A waveguide is placed at the height of $d$ and is parallel to the $x$-axis.
A PA is deployed on the waveguide to serve $\mathrm{U}_b$ with a LoS link.
Since Willie is in the same room with Bob, he also receives a LoS signal from the PA.

\begin{figure}
    \centering
\includegraphics[width=1\linewidth]{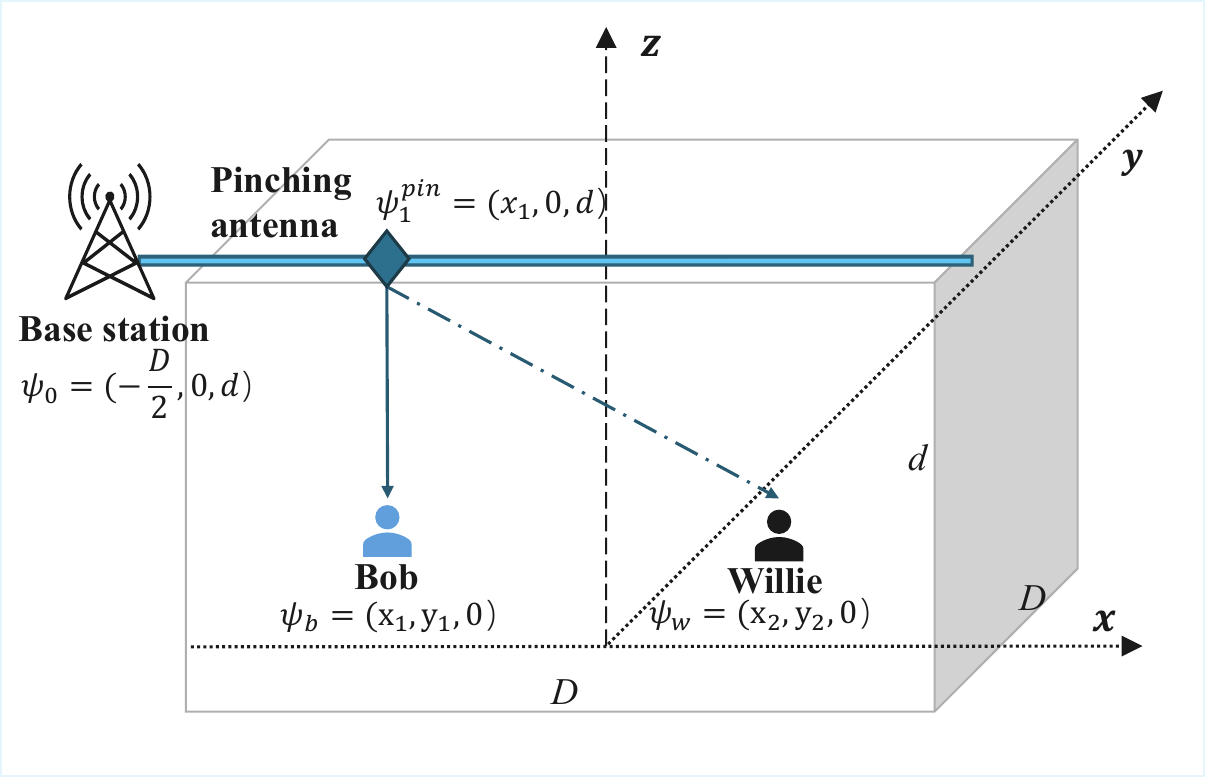}
    \caption{System model consisting of a single-user direct transmission scenario, where a BS serves a legitimate user (Bob) in the presence of an illegitimate user (Willie). }
    \label{fig:system_model}
\end{figure}

To maximize Bob's received signal power, the PA is placed at the position closest to Bob.
Therefore, the $\mathrm{U}_b$'s data rate is given by
\begin{equation}\label{eq-Rb_attenuation}
\begin{split}
R_{b}=\frac{1}{2}\log _2\left( 1+\frac{\eta P_1 e^{-2\alpha\left| \psi_{0} -  \psi _{1}^{Pin} \right|}}{\left| \psi _{1}^{Pin}-\psi _b \right|^2\sigma^2_b} \right),
\end{split}
\end{equation}
where $\eta = \frac{c^2}{16 \pi^2 f_c^2}$, $c$ denotes the speed of light, $f_c$ is the carrier frequency, $P_1$ is the transmit power for the signal of the BS, $\alpha$ is the attenuation constant, $\left|\psi_{0}-\psi_{1}^{Pin}\right|$ is the propagation distance along the waveguide, $\left|\psi_{1}^{Pin}-\psi_{b}\right|$ denotes the distance between the PA and Bob, and $\sigma_b$ and $\sigma_w$ denote the additive white Gaussian noise (AWGN) at $\mathrm{U}_b$ and $\mathrm{U}_w$, respectively.  
Moreover, the data rate of $\mathrm{U}_w$ is given by
\begin{equation}\label{eq-Rw_attenuation}
\begin{split}
R_{w}=\frac{1}{2}\log _2\left( 1+\frac{\eta P_1e^{-2\alpha\left| \psi_{0} -  \psi _{1}^{Pin} \right|}}{\left| \psi _{1}^{Pin}-\psi _w \right|^2\sigma^2_w} \right),
\end{split}
\end{equation}
Then, the secrecy capacity is given by
\begin{equation}\label{eq-r-s-corrected}
\begin{split}
R_{s} = R_{b}-R_{w} = \frac{1}{2}\log _2\left( \frac{1+\frac{\eta P_1 e^{-2\alpha|\psi_0 - \psi_1^{Pin}|}}{\left| \psi _{1}^{Pin}-\psi _b \right|^2\sigma_b ^2}}{1+\frac{\eta P_1 e^{-2\alpha|\psi_0 - \psi_1^{Pin}|}}{\left| \psi _{1}^{Pin}-\psi _w \right|^2\sigma_w ^2}} \right).
\end{split}
\end{equation}

\section{Performance Analysis for PA system}
\subsection{Distance Statistics}
By referring to \cite{cheng_2025_on}, the distance statistics for the proposed PA system are characterized in the following lemmas.
\begin{lemma}\label{lemma-distance-pin-bob}
Denote $Z_b = |\psi_1^{Pin} - \psi_b|^2$, where $\psi_1^{Pin} = \left[ x_1, 0, d \right]$, $\psi_b = \left[x_1, y_1, 0 \right]$, $x_1 \in \left[ -\frac{D}{2}, \frac{D}{2} \right]$ and $y_1 \in \left[ -\frac{D}{2}, \frac{D}{2} \right]$. 
Its PDF and CDF are given by
\begin{equation}\label{lemma-1-pdf}
\begin{split}
f_{Z_b}(z) = \left\{\begin{matrix}
0, & z<d^2, \\
\frac{1}{D\sqrt{z-d^2}}, & d^2\le z\le d^2+\frac{D^2}{4}, \\
0, & d^2+\frac{D^2}{4}\le z,
\end{matrix}\right.
\end{split}
\end{equation}
and
\begin{equation}\label{lemma-1-cdf}
\begin{split}
F_{Z_b}(z) = \left\{\begin{matrix}
0, & z<d^2, \\
\dfrac{2}{D}\sqrt{z-d^2}, & d^2 \leq z \leq d^2 + \tfrac{D^2}{4}, \\
1, & z \geq d^2 + \tfrac{D^2}{4}.
\end{matrix}\right.
\end{split}
\end{equation}
respectively.
\end{lemma}

\begin{lemma}\label{lemma-distance-pin-willie}
Denote $Z_w = |\psi_1^{Pin} - \psi_w|^2$, where $\psi_w = \left[x_2, y_2, 0 \right]$, $x_2 \in \left[ -\frac{D}{2}, \frac{D}{2} \right]$ and $y_2 \in \left[ -\frac{D}{2}, \frac{D}{2} \right]$. 
Its PDF and CDF are given by
\begin{equation}\label{lemma-2-pdf}
\begin{split}
f_{Z_w}(z) = \left\{\begin{matrix}
 0, & z<d^2, \\
 \frac{\pi}{D^2}-\frac{2\sqrt{\zeta}}{D^3}, & d^2\le z\le d^2+\frac{D^2}{4}, \\
 \xi ,& d^2+\frac{D^2}{4}\le z\le d^2 + D^2, \\
 \epsilon,  & d^2 + D^2 \le z\le d^2 + \frac{5D^2}{4},\\
 0, & z\ge d^2 + \frac{5D^2}{4},
\end{matrix}\right.
\end{split}
\end{equation}
and
\begin{equation}\label{lemma-2-cdf}
\begin{split}
F_{Z_w}(z) = \left\{\begin{matrix}
 0, & z<d^2, \\
 \frac{\pi \zeta}{D^2}-\frac{4\zeta ^{\frac{3}{2}}}{3D^3}, & d^2\le z<d^2+\frac{D^2}{4}, \\
 \delta, & d^2+\frac{D^2}{4}\le z\le d^2 + D^2, \\
 \theta, & d^2 + D^2 \le z\le d^2 + \frac{5D^2}{4},\\
 1, & z\ge d^2 + \frac{5D^2}{4}, 
\end{matrix}\right.
\end{split}
\end{equation}
respectively, where $\zeta = z -d^2$,
$\xi = \frac{2}{D^2}\arcsin\left(\frac{D}{2\sqrt{\zeta}}\right) - \frac{1}{D^2}$,
$\epsilon = \frac{2}{D^2}(\arcsin\left(\frac{D}{2\sqrt{\zeta}}\right)-\arcsin\left(\sqrt{1-\frac{D^2}{\zeta}}\right))- \frac{1}{D^2}+\frac{2}{D^3}\sqrt{\zeta-D^2}$,
$\theta = \frac{1}{D}\sqrt{\zeta-\frac{D^2}{4}}+ \frac{2\zeta}{D^2}(\arctan\left(\frac{D}{2\sqrt{\zeta-\frac{D^2}{4}}}\right)- \arctan\left(\frac{\sqrt{\zeta-D^2}}{D}\right))-\frac{\zeta}{D^2}+\frac{1}{12}+\frac{2}{3D^3}\sqrt{\zeta- D^2}(2\zeta+D^2)$,
and 
$\delta = \frac{1}{D}\sqrt{\zeta -\frac{D^2}{4}}+ \frac{2\zeta}{D^2}\arcsin\left(\frac{D}{2\sqrt{\zeta}}\right)- \frac{\zeta}{D^2}+ \frac{1}{12}$.
\end{lemma}

\subsection{Secrecy Outage Probability}
The SOP of the system is given by
\begin{equation}\label{eq-op-bob}
\begin{split}
\mathbb{P} = \Pr\left( R_s < \bar{R} \right),
\end{split}
\end{equation}
where $\bar{R}  =0.01 \text{ bps/Hz}$ denotes the target secrecy rate.
Based on the distance statistics in Lemma~\ref{lemma-distance-pin-bob} and Lemma~\ref{lemma-distance-pin-willie}, the SOP of Bob is derived in the following theorem.
\begin{theorem}\label{SOP Boundary theorem} 
The SOP of Bob is given by
\begin{equation} \label{eq SOP upper and lower}
\begin{split}
\mathbb{P} \begin{cases}\le 1- (\mathcal{J}_1 +  \mathcal{K}_1 + \mathcal{L}_1),\\
\ge 1- (\mathcal{J}_2 +  \mathcal{K}_2 + \mathcal{L}_2),\\
\end{cases}
\end{split}
\end{equation}
where $\mathcal{J}_{1,2}=\sum_{i=1}^n{\omega_i}j_{1,2} (t_i)$,$\mathcal{K}_{1,2}=\sum_{i=1}^n{\omega_i}k_{1,2} (t_i)$ and $\mathcal{L}_{1,2}=\sum_{i=1}^n{\omega_i}l_{1,2} (t_i)$, $n$ is the number of nodes for the Chebyshev-Gauss quadrature, $\omega _i=\frac{\pi}{n}$, and $t_i=\cos \left( \frac{\left( 2i-1 \right) \pi}{2n} \right)$. 
The terms $j_{1,2}(t)$, $k_{1,2}(t)$, and $l_{1,2}(t)$ for the upper and lower bounds are unified by introducing the coefficient pairs $(A_{1,2}, B_{1,2})$. Specifically, we define $(A_1, B_1) = (e^{-2\alpha D}, 1)$ for the upper bound and $(A_2, B_2) = (1, e^{-2\alpha D})$ for the lower bound.
Accordingly, the terms are given as:
$j_{1,2}(t)= F_{Z_{b}}(\frac{\eta \rho A_{1,2}}{4^{\bar{R}} + \frac{4^{\bar{R}} \eta \rho B_{1,2}}{\frac{D^2}{8} t + \frac{D^2}{8} + d^2}-1})(\frac{\pi}{D^2}-\frac{2\sqrt{\frac{D^2}{8} t + \frac{D^2}{8}}}{D^3})\frac{D^2}{8} \sqrt{1-t^2}$,
$k_{1,2}(t)= F_{Z_{b}}(\frac{\eta \rho A_{1,2}}{4^{\bar{R}} + \frac{4^{\bar{R}} \eta \rho B_{1,2}}{\frac{3D^2}{8}t+\frac{5D^2}{8}+d^2}-1})(\frac{2}{D^2}\arcsin\left(\frac{D}{2\sqrt{\frac{3D^2}{8}t+\frac{5D^2}{8}}}\right) - \frac{1}{D^2})\frac{3D^2}{8} \sqrt{1-t^2}$, and
$l_{1,2}(t)= \frac{D^2}{8}\sqrt{1-t^2}F_{Z_{b}}(\frac{\eta \rho A_{1,2}}{4^{\bar{R}} + \frac{4^{\bar{R}} \eta \rho B_{1,2}}{\frac{D^2}{8}t + \frac{9D^2}{8} + d^2}-1})(\frac{2}{D^2}(\arcsin\left(\frac{2}{\sqrt{t+9}}\right)- \arcsin\left( \frac{t+1}{t+9} \right))- \frac{1}{D^2}+\frac{1}{D^2}\sqrt{\frac{t+1}{2}})$.
\end{theorem}

\begin{proof}
Based on $0 \le \left| \psi_{0} -  \psi _{1}^{\mathrm{Pin}} \right| \le D$,
the attenuation factor $e^{-2 \alpha \left| \psi_{0} -  \psi _{1}^{\mathrm{Pin}} \right|}$
is upper bounded by $1$ and lower bounded by $e^{-2 \alpha D}$.

Thus, the upper and lower bounds of $R_b$ are given by
\begin{equation}\label{eq-bob-upper}
R_{b}^{\mathrm{upper}}
=\frac{1}{2}\log _2\left( 1+\frac{\eta P_1}{\left| \psi _{1}^{\mathrm{Pin}}-\psi _b \right|^2\sigma_b^2} \right),
\end{equation}
and
\begin{equation}\label{eq-bob-lower}
R_{b}^{\mathrm{lower}}
=\frac{1}{2}\log _2\left( 1+\frac{\eta P_1 e^{-2 \alpha D}}{\left| \psi _{1}^{\mathrm{Pin}}-\psi _b \right|^2\sigma_b^2} \right),
\end{equation}
respectively. Similarly, the upper and lower bounds of $R_w$ are given by
\begin{equation}\label{eq-willie-upper}
R_{w}^{\mathrm{upper}}
=\frac{1}{2}\log _2\left( 1+\frac{\eta P_1}{\left| \psi _{1}^{\mathrm{Pin}}-\psi _w \right|^2\sigma_w^2} \right),
\end{equation}
and
\begin{equation}\label{eq-willie-lower}
R_{w}^{\mathrm{lower}}
=\frac{1}{2}\log _2\left( 1+\frac{\eta P_1 e^{-2 \alpha D}}{\left| \psi _{1}^{\mathrm{Pin}}-\psi _w \right|^2\sigma_w^2} \right),
\end{equation}
respectively. 
Hence,
$R_{s}^{\mathrm{upper}} = R_{b}^{\mathrm{upper}}-R_{w}^{\mathrm{lower}}$
and
$R_{s}^{\mathrm{lower}} = R_{b}^{\mathrm{lower}}-R_{w}^{\mathrm{upper}}$.

To derive the upper bound on the SOP, we consider the lower bound on the secrecy rate $R_{s}^{\mathrm{lower}} = R_{b}^{\mathrm{lower}} - R_{w}^{\mathrm{upper}}$.
By directly substituting ~\eqref{eq-Rb_attenuation} and ~\eqref{eq-Rw_attenuation}, we obtain the necessary condition for successful transmission: $Z_b \le \frac{\eta \rho e^{-2\alpha D}}{4^{\bar{R}} \left( 1 + \frac{\eta \rho}{Z_w} \right) - 1}$. 
Based on ~\eqref{lemma-1-cdf} and ~\eqref{lemma-2-pdf}, Chebyshev–Gauss quadrature can be applied for each integration. We change the integration domain into $[-1,1]$ and let $z =\frac{D^2}8t + \frac{D^2}8 + d^2$, then we have
\begin{equation}
\begin{aligned}
\mathcal{J}_{1} &= \int_{-1}^{1} \frac{D^2}{8}  F_{Z_{b}}( \frac{\eta \rho e^{-2\alpha D}}{4^{\bar{R}} + \frac{4^{\bar{R}} \eta \rho }{\frac{D^2}{8} t + \frac{D^2}{8} + d^2}-1} ) \\&\times  \left( \frac{\pi}{D^2}-\frac{2\sqrt{\frac{D^2}{8} t + \frac{D^2}{8}}}{D^3} \right) dt \approx \sum_{i=1}^{n}\omega_i j_1(t_i),
\end{aligned}
\end{equation}
Similarly with $\mathcal{J}_{2}$, $\mathcal{K}_{1,2}$ and $\mathcal{L}_{1,2}$, the proof has been completed.
\end{proof}
\begin{proposition}\label{proposition-high-op}
In the high-SNR regime, $\mathbb{P}$ can be approximated as
\begin{equation}\begin{split}
\mathbb{P} \begin{cases}\le 1- \left(\mathcal{J}_1^{\infty} +  \mathcal{K}_{1}^{\infty} + \mathcal{L}_{1}^{\infty}\right),\\
\ge 1- \left(\mathcal{J}_{2}^{\infty} +  \mathcal{K}_{2}^{\infty} + \mathcal{L}_{2}^{\infty}\right),
\end{cases}\end{split}\end{equation}
\end{proposition}
\begin{proof}
When $\rho \to \infty$, we can obtain $\lim_{\rho \to \infty} \frac{\eta \rho e^{-2\alpha D}}{4^{\bar{R}} + \frac{4^{\bar{R}} \eta \rho }{Z_{w}}-1} \approx \frac{Z_{W}e^{-2\alpha D}} {4^{\bar{R}}}$ and  $\lim_{\rho \to \infty} \frac{\eta \rho}{4^{\bar{R}} + \frac{4^{\bar{R}} \eta \rho e^{-2\alpha D}}{Z_{w}}-1} \approx \frac{Z_{W}} {4^{\bar{R}}e^{-2\alpha D}}$, so $ \mathcal{J}_{1}^{\infty}$ can be derived by
\begin{align}
\mathcal{J}_{1}^{\infty} & = \int_{-1}^{1} \frac{D^2}{8} F_{Z_{b}}\left( \frac{Z_{W}e^{-2\alpha D}} {4^{\bar{R}}} \right)  \\&\times \left( \frac{\pi}{D^2}-\frac{2\sqrt{\frac{D^2}{8} t + \frac{D^2}{8}}}{D^3} \right) dt =\sum_{i=1}^{n}\omega_i j_1^{\infty}(t_i).
\end{align}
Similarly with $\mathcal{J}_{2}^{\infty}$, $\mathcal{K}_{1,2}^{\infty}$ and $\mathcal{L}_{1,2}^{\infty}$. 
The proof has been completed.
\end{proof}

\begin{corollary}\label{corollary-op}
In this system, the diversity order of Bob is given by $\mathcal{D} = 0$.
\end{corollary}
\begin{proof}
The diversity order is defined as $\mathcal{D} = -\lim_{\rho \to \infty} \frac{\log \mathbb{P}}{\log \rho}$.
Based on the derived bounds in Theorem~\ref{SOP Boundary  theorem}, we have $-\lim_{\rho \to \infty} \frac{\log \mathbb{P}_{LB}}{\log \rho} = -\lim_{\rho \to \infty} \frac{\log \mathbb{P}_{UB}}{\log \rho} = 0$.
Thus, we can obtain that $\mathcal{D} = 0$. This completes the proof.
\end{proof}

\subsection{Ergodic Secrecy Capacity}
\begin{theorem}
The ESC fort Bob is given by
\begin{equation}
\begin{split}
\mathbb{E}(R_{s}) \left\{\begin{matrix}
\le \mathcal{C}_3 - \mathcal{J}_3 - \mathcal{K}_3 - \mathcal{L}_3, \\
\ge \mathcal{C}_4 - \mathcal{J}_4 - \mathcal{K}_4 - \mathcal{L}_4,
\end{matrix}\right.
\end{split}
\end{equation}
where 
$\mathcal{C}_{3,4} = \sum_{i=1}^{n}w_ic_{3,4}(t_i)$,
$\mathcal{J}_{3,4} = \sum_{i=1}^{n}w_ij_{3,4}(t_i)$, 
$\mathcal{K}_{3,4} = \sum_{i=1}^{n}w_ik_{3,4}(t_i)$, 
$\mathcal{L}_{3,4} = \sum_{i=1}^{n}w_il_{3,4}(t_i)$. 
The terms $c_{3,4}(t)$, $j_{3,4}(t)$, $k_{3,4}(t)$, and $l_{3,4}(t)$ are unified by introducing the coefficient pairs $(A_{3,4}, B_{3,4})$. We define $(A_3, B_3) = (1, e^{-2\alpha D})$ for the upper bound and $(A_4, B_4) = (e^{-2\alpha D}, 1)$ for the lower bound. Accordingly, the terms are given as:
$c_{3,4}(t) = \frac{D^2}{8}\log_2(1+\frac{\eta \rho A_{3,4}}{\frac{D^2}{8} t + \frac{D^2}{8} + d^2}) \frac{1}{D\sqrt{\frac{D^2}{8} t + \frac{D^2}{8}}} \sqrt{1-t^2}$, 
$j_{3,4}(t)=\frac{D^2}{8} \log_2(1+\frac{\eta \rho B_{3,4}}{\frac{D^2}{8}t + \frac{D^2}{8} + d^2})(\frac{\pi}{D^2}-\frac{2\sqrt{\frac{D^2}{8}t+\frac{D^2}{8}}}{D^3})\sqrt{1-t^2}$, 
$k_{3,4}(t)=\frac{3D^2}{8} \log_2(1+\frac{\eta\rho B_{3,4}}{\frac{3D^2}{8} t+\frac{5D^2}{8}+d^2})(\frac{2}{D^2}\arcsin(\frac{D}{2\sqrt{\frac{3D^2}{8} t+\frac{5D^2}{8}}})-\frac{1}{D^2})\sqrt{1-t^2}$, and 
$l_{3,4}(t)=\frac{D^2}{8}\log_2(1+\frac{\eta\rho B_{3,4}}{\frac{D^2}{8} t+\frac{9D^2}{8}+d^2})\epsilon\sqrt{1-t^2}$.
\end{theorem}

\begin{proof}
Based on $0 \le \left| \psi_{0} -  \psi _{1}^{Pin} \right| \le D$, we can obtain that the upper bound and lower bound of $e^{-2 \alpha \left| \psi_{0} -  \psi _{1}^{Pin} \right|}$ are $1$ and $e^{-2 \alpha D}$, respectively.

According to \eqref{eq-bob-upper}, \eqref{eq-bob-lower}, \eqref{eq-willie-upper} and \eqref{eq-willie-lower}, the upper and lower bounds of secrecy capacity are $R_{s}^{upper} = R_{b}^{upper}-R_{w}^{lower}$ and $R_{s}^{lower}=R_{b}^{lower}-R_{w}^{upper}$, respectively.

Next, the upper and lower bounds of ESC are given by $\mathbb{E}(R_{s}^{upper}) = \mathbb{E}(R_{b}^{upper}) - \mathbb{E}(R_{w}^{lower})$ and $\mathbb{E}(R_{s}^{lower}) = \mathbb{E}(R_{b}^{lower}) - \mathbb{E}(R_{w}^{upper})$, respectively.

To perform Chebyshev–Gauss quadrature for each integration, we need a variable substitution to change the integration to $[-1,1]$, i.e. 
Let $z = \frac{D^2}{8} t + \frac{D^2}{8} + d^2$, and we have

\begin{equation}
\begin{split}
\mathbb{E}(R_{b}^{upper})
&= \int_{-1}^{1} \frac{D^2}{8} \log_2\left(1 + \frac{\eta \rho}{\frac{D^2}{8} t + \frac{D^2}{8} + d^2}\right)f_{Z_{b}}(t)\, dt \\
&\approx \sum_{i=1}^{n}w_{i}j_3(t_i),
\end{split}
\end{equation}
Similarly, $\mathbb{E}(R_{w}^{lower})$, $\mathbb{E}(R_{b}^{lower})$, and $\mathbb{E}(R_{w}^{upper})$ can be solved in the same way.
The proof is complete.
\end{proof}

To further characterize system performance, we define the high-SNR slope as $\mathcal{S}=\lim_{\rho \to \infty} \frac{R(\rho)}{\log_2(\rho)}$.
Its evaluation relies on the asymptotic expression for Bob's ESC, which is presented in the following corollary.
\begin{corollary}
In the high-SNR regime, the approximation of $\mathbb{E}(R_s)$ is given by
\begin{equation}
\begin{split}
\mathbb{E}(R_s)^{\infty} \left\{\begin{matrix}
\le \mathcal{T}_1 - \mathcal{T}_2- \log_2(e^{-2 \alpha D}), \\
\ge \mathcal{T}_1 - \mathcal{T}_2 +\log_2(e^{-2 \alpha D}),
\end{matrix}\right.
\end{split}
\end{equation}
where $\mathcal{T}_1 = \sum_{i=1}^n w_i j_5(t_i) + \sum_{i=1}^n w_i k_5(t_i) + \sum_{i=1}^n w_i l_5(t_i)$,
$\mathcal{T}_2 = \sum_{i=1}^n w_i c_5(t_i)$,
$c_5(t) = \frac{D^2}{8}\log_2(\frac{D^2}{8} t + \frac{D^2}{8} + d^2) \frac{1}{D\sqrt{\frac{D^2}{8} t + \frac{D^2}{8}}} \sqrt{1-t^2}$, 
$j_5(t)=\frac{D^2}{8} \log_2(\frac{D^2}{8}t + \frac{D^2}{8} + d^2)(\frac{\pi}{D^2}-\frac{2\sqrt{\frac{D^2}{8}t+\frac{D^2}{8}}}{D^3})\sqrt{1-t^2}$,
$k_5(t)=\frac{3D^2}{8} \log_2(\frac{3D^2}{8} t+\frac{5D^2}{8}+d^2)(\frac{2}{D^2}\arcsin(\frac{D}{2\sqrt{\frac{3D^2}{8} t+\frac{5D^2}{8}}})-\frac{1}{D^2})\sqrt{1-t^2}$, and
$l_5(t)=\frac{D^2}{8}\log_2(\frac{D^2}{8} t+\frac{9D^2}{8}+d^2)\epsilon\sqrt{1-t^2}$.
\end{corollary}

\begin{proof}
When $\rho \to \infty$, we can obtain $\lim_{\rho \to \infty} \log_2\left(\frac{1+ \frac{\eta \rho}{x}
}{1+\frac{\eta \rho e^{-2\alpha D}}{y}}\right) \approx \log_2\left(\frac{y}{x}\right) - \log_2(e^{-2\alpha D})$ and  $\lim_{\rho \to \infty} \log_2\left(\frac{1+ \frac{\eta \rho e^{-2\alpha D}}{x}
}{1+\frac{\eta \rho }{y}}\right) \approx \log_2\left(\frac{y}{x}\right) + \log_2(e^{-2\alpha D})$, $\mathbb{E}(R_{s}^{upper})^{\infty}$ can be derived as
\begin{equation}
\begin{split}
\mathbb{E}(R_{s}^{upper})^{\infty} &= \int_{0}^{\infty}\int_{0}^{\infty} \log_2 \left(\frac{y}{x}\right) f_{Z_{b}}(x) f_{Z_{w}}(y) dxdy  \\ &- \log_2(e^{-2 \alpha D}) \\
& = \underbrace{\int_{0}^{\infty} \log_2 \left(y\right)f_{Z_{w}}(y)dy}_{\mathcal{T}_1} \\
&- \underbrace{\int_{0}^{\infty}\log_2 \left(x\right) f_{Z_{b}}(x)dx}_{\mathcal{T}_2} - \log_2(e^{-2 \alpha D}),
\end{split}
\end{equation}
Next, $\mathcal{T}_2$ can be estimated through the Chebyshev–Gauss quadrature approach, and then we have
\begin{equation}
\begin{split}
    \mathcal{T}_2 &= \int_{-1}^{1} \frac{D^2}{8} \log_2\left(\frac{D^2}{8} t + \frac{D^2}{8} + d^2\right)f_{Z_{b}}(t) dt \\
    & \approx \sum_{i = 1}^{n} w_i c_5(t_i),
\end{split}
\end{equation}
Similarly, $\mathcal{T}_1$ and $\mathbb{E}(R_{s}^{lower})^{\infty}$ can be solved in the same way.
This completes the proof.
\end{proof}

\begin{corollary}\label{corollary-er}
In this system, the high-SNR slope of Bob is given by $\mathcal{S} = 0$.
\end{corollary}
\begin{proof}
In this system, we have $\mathcal{S} = \frac{dR^{\infty}}{d\log_2(\rho)}$.
According to $\frac{dR^{\infty,upper}}{d\log_2(\rho)} = \frac{dR^{\infty,lower}}{d\log_2(\rho)} = 0$, we can obtain that $\mathcal{S} = 0$.
This completes the proof.
\end{proof}

\section{Numerical Results}
In this section, numerical results are provided to evaluate the performance of the proposed attenuation model. 
Meanwhile, Monte-Carlo simulations are used to validate the accuracy of the derived analytical expressions. 
Table~\ref{tab:parameters} shows specific parameter settings.
\begin{table}[htbp]
\centering
\caption{Parameters Settings}
\begin{tabular}{|c|c|}
\hline
Side length of the region & $D =25$ m \\
\hline
Height of the waveguide & $d = 3$ m \\
\hline
Carrier frequency & $f_c = 10$ GHz \\
\hline
Waveguide attenuation coefficient & $\alpha = 0.01$ \\
\hline
Number of points for Chebyshev-Gauss quadratures & $N = 1000$ \\
\hline
Bandwidth & $  B =1 $ MHz  \\
\hline
Target rates & $R=10 $ kbps \\
\hline
Monte Carlo simulations & $5 \times 10^4$ \\
\hline
\end{tabular}
\label{tab:parameters}
\end{table}

\subsection{Secrecy Outage Probability Analysis}
The SOPs versus the transmit SNR are illustrated in Fig.~\ref{fig:op_low} and Fig.~\ref{fig:op_high}.
As observed, the Monte Carlo simulation results are strictly bracketed by the derived analytical upper and lower bounds across the entire SNR range. 
This consistency validates the accuracy of the proposed mathematical framework and the derived analytical expressions.

\begin{figure*}[t]
    \centering

    \begin{minipage}[t]{0.24\textwidth}
        \centering
        \includegraphics[width=\linewidth]{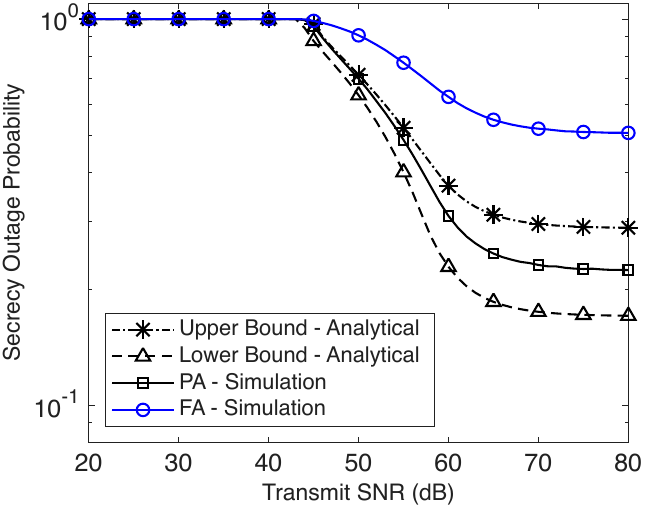}
        \captionof{figure}{SOPs of the PA and FA in the low-SNR regime.}
        \label{fig:op_low}
    \end{minipage}
    \hfill
    \begin{minipage}[t]{0.24\textwidth}
        \centering
        \includegraphics[width=\linewidth]{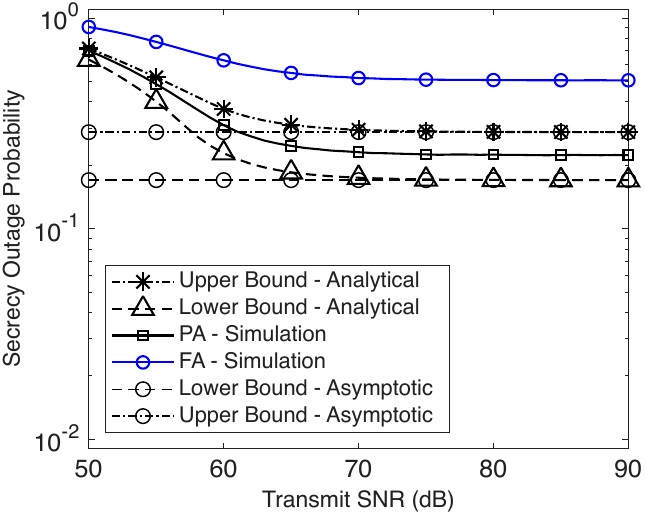}
        \captionof{figure}{SOPs of the PA and FA in the high-SNR regime.}
        \label{fig:op_high}
    \end{minipage}
    \hfill
    \begin{minipage}[t]{0.24\textwidth}
        \centering
        \includegraphics[width=\linewidth]{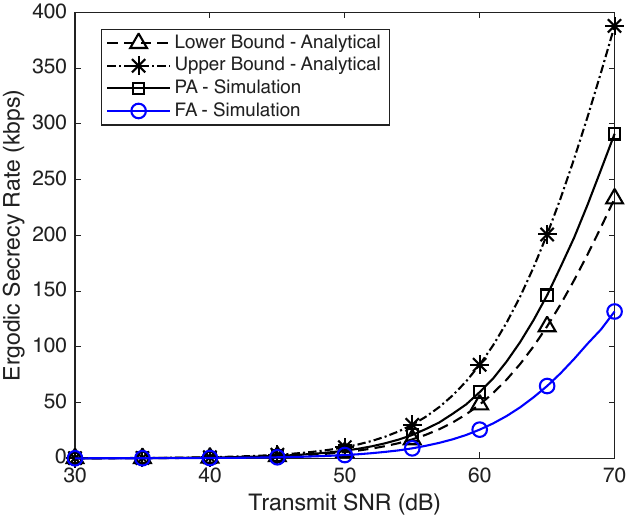}
        \captionof{figure}{ESCs of the PA and FA in the low-SNR regime.}
        \label{fig:er_low}
    \end{minipage}
    \hfill
    \begin{minipage}[t]{0.24\textwidth}
        \centering
        \includegraphics[width=\linewidth]{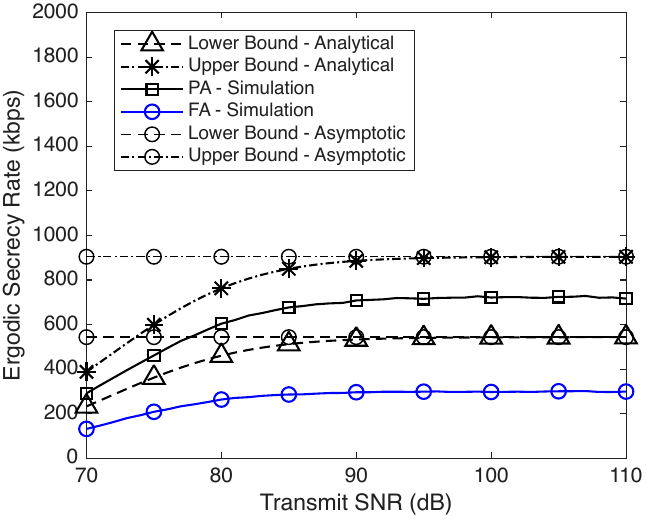}
        \captionof{figure}{ESCs of the PA and PA in the high-SNR regime.}
        \label{fig:er_high}
    \end{minipage}

\end{figure*}
Furthermore, the analytical bounds converge to the derived asymptotic limits in the high-SNR regime, substantiating the validity of the asymptotic analysis. 
This saturation behavior confirms that the system achieves a diversity order of $\mathcal{D}=0$, which is consistent with the theoretical results in Corollary~\ref{corollary-op}.
For performance comparison, we consider a conventional FA system as the benchmark, which is located at the signal source $[0,0,d]$ and is unaffected by in-waveguide attenuation.
It is observed that the PA has a consistently lower SOP than the FA across the evaluated SNR range, 
which reveals that the PA has higher reliability for secure communication even when accounting for in-waveguide attenuation.
Specifically, the spatial flexibility of the PA allows
an adjustable radiating point along the waveguide to find more favorable channel conditions, 
effectively mitigating the impact of large-scale path loss.

\subsection{Ergodic Secrecy Capacity Analysis}
In Fig.~\ref{fig:er_low} and Fig.~\ref{fig:er_high}, the ESCs versus the transmit SNR with PA and FA are plotted.
Initially, the simulation points are seen to align well with the analytical results, being strictly bracketed by the analytical lower and upper bounds derived from the theoretical analysis.
It is observed that the analytical upper bound provides a particularly tight approximation to the simulation results across the entire SNR range.
Furthermore, the validity of the asymptotic analysis is substantiated as ESC curves converge to the asymptotic lower and upper bounds in the high-SNR regime.
The analysis reveals that the ESC attains a slope of $0$ in the high-SNR regime, which is consistent with Corollary~\ref{corollary-er}.
Additionally, the simulation curves confirm that the PA system consistently maintains a higher ESC compared to the traditional FA baseline,
indicating a superior capacity for secure data transmission against eavesdropping.

\section{Conclusions}
In this paper, we have investigated the PLS performance of the PA system under practical in-waveguide attenuation. 
Our asymptotic analysis has revealed that the system exhibits a zero diversity order and a high-SNR slope of zero, 
confirming that the secrecy performance saturates in the high-SNR regime. 
Numerical results have validated the tightness of the derived bounds and demonstrated that the secrecy reliability is fundamentally geometry-limited due to the coupling between waveguide attenuation and activation uncertainty. 
Furthermore, the derived closed-form expressions for the SOP and ESC could facilitate the future PA systems evaluation. 
By eliminating the high computational costs of exhaustive simulations, these analytical results facilitate antenna selection in multi-PA systems and the optimization of deployment strategies to enhance secrecy performance in attenuation-aware environments.

\bibliographystyle{IEEEtran}
\bibliography{reference}
\end{document}